\newcounter{myctr}
\def\myitem{\refstepcounter{myctr}\bibfont\noindent\ifnum\themyctr>9\else\phantom{0}\fi\hangindent17pt\themyctr.\enskip}

\documentclass{ws-ijqi}
\usepackage{amsmath}
\usepackage{longtable,mathdots}
\usepackage{latexsym,enumerate,amssymb,amsbsy}
\usepackage{verbatim,color}
\usepackage{supertabular}
\usepackage{algorithm,algorithmic}
\usepackage{breakurl}
\usepackage{cite}
\newtheorem{conjecture}[theorem]{Conjecture}
\numberwithin{equation}{section}

\newcommand{\C}{\mathbb C}
\newcommand{\F}{\mathbb{F}}

\newcommand{\GRS}{\mathcal{GRS}}
\newcommand{\etal}{\emph{et al.}}
\def\wt{\mathop{{\rm wt}}}

\def\E{\mathop{{\rm E}}}
\def\u{{\mathbf{u}}}
\def\al{{\boldsymbol{\alpha}}}
\def\v{{\mathbf{v}}}

\def\0{{\mathbf{0}}}

\begin{document}

\markboth{Ezerman, Jitman, Kiah, Ling}
{Pure Asymmetric Quantum MDS Codes from CSS Construction: A Complete Characterization}


\title{\bf{PURE ASYMMETRIC QUANTUM MDS CODES\\
FROM CSS CONSTRUCTION:\\
A COMPLETE CHARACTERIZATION}}

\author{MARTIANUS FREDERIC EZERMAN}

\address{Centre for Quantum Technologies, National University of Singapore\\
Block S15, 3 Science Drive 2, Singapore 117543, Republic of Singapore\\
cqtmfe@nus.edu.sg, frederic.ezerman@gmail.com}

\author{SOMPHONG JITMAN}

\address{Division of Mathematical Sciences, School of Physical and Mathematical Sciences,\\
Nanyang Technological University\\ 
21 Nanyang Link, Singapore 637371, Republic of Singapore\\
sjitman@ntu.edu.sg}

\author{KIAH HAN MAO}

\address{Division of Mathematical Sciences, School of Physical and Mathematical Sciences,\\
Nanyang Technological University\\ 
21 Nanyang Link, Singapore 637371, Republic of Singapore\\
kiah0001@ntu.edu.sg}

\author{SAN LING}

\address{Division of Mathematical Sciences, School of Physical and Mathematical Sciences,\\
Nanyang Technological University\\ 
21 Nanyang Link, Singapore 637371, Republic of Singapore\\
lingsan@ntu.edu.sg}

\maketitle
\begin{abstract}
Using the Calderbank-Shor-Steane (CSS) construction, pure $q$-ary asymmetric quantum error-correcting codes 
attaining the quantum Singleton bound are constructed. Such codes are called pure CSS asymmetric quantum maximum distance 
separable (AQMDS) codes. Assuming the validity of the classical MDS Conjecture, pure CSS AQMDS codes of all possible 
parameters are accounted for.
\end{abstract}


\section{Introduction}\label{sec:intro}
The study of {\it asymmetric quantum codes (AQCs)} began  when it was argued that, in many qubit systems, 
phase-flips (or Z-errors) occur more frequently than bit-flips (or X-errors) do~\cite{ESCH07,IM07}. 
Steane first hinted the idea of adjusting the error-correction to the particular characteristics of the 
quantum channel in~\cite{Ste96} and later, Wang \etal~established a mathematical model of AQCs in the 
general qudit system in~\cite{WFLX09}.

To date, the only known  class of AQCs is given by the asymmetric version of the CSS construction~\cite{AA10,WFLX09}. 
In this paper, the CSS construction is used to derive a class of pure\footnote{Purity in the CSS case is defined in 
Theorem~\ref{thm:main}} AQCs  attaining the quantum analogue of the Singleton bound. We call such optimal 
codes {\it asymmetric quantum maximum distance separable (AQMDS)} codes and if the codes are derived 
from the CSS construction, we call them CSS AQMDS codes.

Thus far, the only known AQMDS codes are pure CSS AQMDS and many results concerning these codes had been discussed in~\cite{EJL2011}. 
This paper provides a complete treatment of such codes by solving the remaining open problems. This enables us to provide 
a complete characterization. To be precise, assuming the validity of the MDS conjecture, pure CSS AQMDS codes of all 
possible parameters are constructed.

The paper is organized as follows. In Section~\ref{sec:prelim}, we discuss some preliminary concepts and results. 
In Sections~\ref{sec:GRSuptoq} to~\ref{sec:n=q+2}, nested pairs of Generalized Reed-Solomon (GRS) codes and 
extended GRS codes are used to derive AQMDS codes of lengths up to $q+2$. Section~\ref{sec:dxequals2} presents 
an alternative view on the construction of AQMDS codes based on the weights of MDS codes. A summary is provided in Section~\ref{sec:summary}.

\section{Preliminaries}\label{sec:prelim}
\subsection{Classical linear MDS codes}
Let $q$ be a prime power and ${\F}_{q}$ the finite field having 
$q$ elements. A {\it linear $[n,k,d]_q$-code} $C$ is a $k$-dimensional ${\F}_{q}$-subspace of ${\F}_{q}^n$ with 
{\em minimum distance} $d:=\min \{ \wt(\v): \v \in C \setminus \{\0 \}\}$, where  $\wt(\v)$  denotes the 
{\it Hamming weight}  of $\v \in {\F}_{q}^{n}$. Given two distinct linear codes $C$ and $D$, $\wt(C \setminus D)$ denotes 
$\min \{ \wt(\u) :\u\in C \setminus D \}$. 
 
Every  $[n,k,d]_{q}$-code $C$ satisfies the Singleton bound 
$$ d\leq n-k+1 \text{,}$$
and $C$ is said to be \textit{maximum distance separable (MDS)} if $d=n-k+1$. Trivial families of MDS codes 
include the vector space ${\F}_{q}^{n}$, the codes equivalent to the $[n,1,n]_q$-repetition code and 
their duals $[n,n-1,2]_q$ for positive integers $n \geq 2$. 

MDS codes which are not equivalent to the trivial ones are said to be \textit{nontrivial}. 
Furthermore, we have the following conjecture which has been shown to be true when $q$ is prime in~\cite{Bal10}.

\begin{conjecture}[MDS Conjecture]\label{MDSconjecture}
If there is a nontrivial $[n,k,d]_{q}$-MDS code, then $n \leq q+1$, 
except when $q$ is even and $k=3$ or $k=q-1$ in which case $n \leq q+2$. 
\end{conjecture} 

For $\u=(u_i)_{i=1}^n$ and $\v=(v_i)_{i=1}^n$,
$\langle \u, \v \rangle_{\E}:=\sum_{i=1}^{n} u_i v_i$ is the \textit{Euclidean inner product} of 
$\u$ and $\v$. With respect to this inner product, the \emph{dual} $C^{\perp}$ of $C$ is given by 
\[C^\perp:=\left\lbrace \u \in {\F}_q^n : \left\langle \u,\v\right\rangle_{\E} = 0 
\text{ for all } \v \in C \right\rbrace.\] It is well known that $\left(C^{\perp}\right)^{\perp} = C$ 
and that the dual of an MDS code is MDS.

Let ${\F}_{q}[X]_{k}$ denote the set of all polynomials of degree less than $k$ in ${\F}_{q}[X]$. 
The set $\{1,x,\ldots,x^{k-1}\}$ forms the standard basis for ${\F}_{q}[X]_{k}$ as a vector space over $\F_{q}$.

\subsection{CSS construction and AQMDS codes}
We begin with a formal definition of an AQC.
\begin{definition}\label{def:AQCs}
Let $d_{x}$ and $d_{z}$ be positive integers. A quantum code $Q$ in 
$V_{n}=({\C}^{q})^{\otimes n}$ with dimension $K \geq 1$ is called an 
\textit{asymmetric quantum code} with parameters 
$((n,K,d_{z}/d_{x}))_{q}$ or $[[n,k,d_{z}/d_{x}]]_{q}$, 
where $k=\log_{q}K$, 
if $Q$ detects $d_{x}-1$ qudits of bit-flips (or $X$-errors) and, 
at the same time, $d_{z}-1$ qudits of phase-flips (or $Z$-errors).
\end{definition}

The correspondence between pairs of classical linear codes and AQCs is given in~\cite{AA10,WFLX09}.
\begin{theorem}[Standard CSS Construction for AQC]\label{thm:main}
Let $C_{i}$ be linear codes with parameters $[n,k_{i},d_{i}]_{q}$ for $i=1,2$ 
with $C_{1}^{\perp}\subseteq C_{2}$. Let 
\begin{align}\label{eq:distances}
d_{z}&:=\max \{ \wt(C_{2} \setminus C_{1}^{\perp}), \wt (C_{1} \setminus C_{2}^{\perp}) \} \text{ and }\notag\\ 
d_{x}&:=\min \{ \wt(C_{2} \setminus C_{1}^{\perp}), \wt(C_{1} \setminus C_{2}^{\perp}) \}\text{.}
\end{align}
Then there exists an AQC $Q$ with parameters $[[n,k_{1}+k_{2}-n,d_{z}/d_{x}]]_{q}$. The code $Q$ is said to be 
\textit{pure} whenever $\{ d_{z},d_{x} \} = \{ d_{1},d_{2} \}$.
\end{theorem}

For a CSS AQC, the purity in Theorem~\ref{thm:main} is equivalent to the general definition given in \cite{WFLX09}.

Furthermore, any CSS $[[n,k,d_{z}/d_{x}]]_{q}$-AQC satisfies the following bound~\cite[Lem. 3.3]{SRK09},
\begin{equation}\label{eq:QSB}
k \leq n-d_{x}-d_{z}+2 \text{.}
\end{equation}
This bound is conjectured to hold for all AQCs. A quantum code is said to be
\textit{asymmetric quantum MDS (AQMDS)} if it attains the equality in (\ref{eq:QSB}).

For our construction, the following result  holds.
\begin{lemma}[{\cite[Cor. 2.5]{WFLX09}}]  \label{lem:AQMDS}
A pure CSS AQC is an asymmetric quantum MDS code if and only if both $C_{1}$ and $C_{2}$ in 
Theorem~\ref{thm:main} are (classical) MDS codes.
\end{lemma}
This means that constructing a pure $q$-ary CSS AQMDS code of a specific set of parameters is equivalent 
to finding a suitable corresponding nested pair of classical ${\F}_{q}$-linear MDS codes.

Following Lemma~\ref{lem:AQMDS}, a CSS AQMDS code is said to be \emph{trivial} if both 
$C_{1}$ and $C_{2}$ are trivial MDS codes.

From Lemma~\ref{lem:AQMDS} and the MDS Conjecture, the following 
necessary condition for the existence of a nontrivial pure CSS AQMDS code is immediate.

\begin{proposition}\label{prop:MDSlength}
Assuming the validity of the MDS Conjecture, every nontrivial pure $q$-ary
CSS AQMDS code has length $ n \leq q+1$ if $q$ is odd and $n \leq q+2$ if $q$ is even.
\end{proposition}

Let $Q$ be an AQC with parameters $[[n,k,d_{z}/d_{x}]]_{q}$. We usually require $k>0$ 
(equivalently, $K=q^k>1$) or for error detection purposes, $d_{x} \geq 2$. However, for completeness, 
we state the results for the two cases:  first, when $d_{x}=1$ and second, when $k=0$.

\begin{proposition}\label{prop:dxis1}
Let $n, k$ be positive integers such that $k\le n-1$. A pure CSS AQMDS code with parameters $[[n,k,d_z/1]]_q$ 
where $d_z=n-k+1$ exists if and only if there exists an MDS code with parameters $[n,k,n-k+1]_q$.
\end{proposition}
\begin{proof}
We show only one direction.
Let $C$ be an MDS code with parameters $[n,k,n-k+1]_{q}$. 
Apply Theorem~\ref{thm:main} with $C_{1}=C$ and $C_{2}={\F}_{q}^{n}$ to obtain the required AQMDS code.
\end{proof}

\begin{proposition}\label{prop:kis0}
Let $n,k$ be positive integers such that $k\le n-1$. A pure CSS AQMDS code with parameters $[[n,0,d_z/d_x]]_q$ 
where $\{d_z,d_x\}=\{n-k+1,k+1\}$ exists if and only if there exists an MDS code with parameters $[n,k,n-k+1]_q$.
\end{proposition}

\begin{proof}
Again, we show one direction.
Let $C$ be an MDS code with parameters $[n,k,n-k+1]_{q}$ and let $C_{1}^{\perp}=C_{2}=C$. Following \cite{CRSS98}, assume that a quantum 
code with $K=1$ is pure and hence, there exists an AQMDS 
with parameters $[[n,0,d_{z}/d_{x}]]_{q}$ where $\{ d_{z},d_{x} \} = \{ n-k+1, k+1\}$.
\end{proof}

In the subsequent sections, pure CSS AQMDS codes with $k\ge 1$ and $d_x\ge 2$ are studied.

\section{AQMDS Codes of length $n\leq q$}\label{sec:GRSuptoq}
Let us recall some basic results on GRS codes (see~\cite[Sect. 5.3]{HP03}). 
Choose $n$ distinct elements $\alpha_{1},\alpha_2,\ldots,\alpha_{n}$ in ${\F}_{q}$ 
and define $\al:=(\alpha_{1},\alpha_2,\ldots,\alpha_{n})$. 
Let $\v:=(v_{1},v_2,\ldots,v_{n})$, where $v_{1},v_2,\ldots v_n$ are nonzero elements in ${\F}_{q}$. 
Then, given $\al$ and $\v$, a GRS code of length $n \leq q$ and dimension $k \leq n$ is defined as 
\begin{equation*}
{\GRS}_{n,k}(\al,\v):=
\left\{(v_{1}f(\alpha_{1}),\ldots,v_{n}f(\alpha_{n})):  f(X) \in {\F}_{q}[X]_{k}\right\}\text{.}
\end{equation*}

Since ${\F}_{q}[X]_{k} \subset {\F}_{q}[X]_{k+1}$ 
for fixed $n,\v$, and $\al$, it follows immediately that
\begin{equation}
{\GRS}_{n,k}(\al,\v) \subset {\GRS}_{n,k+1}(\al,\v)\text{.}
\end{equation}

Based on the standard basis for $\F_{q}[X]_{k}$, a generator matrix $G$ for ${\GRS}_{n,k}(\al,\v)$ is given by
\begin{equation}\label{eq:G}
G=\left(
\begin{array}{*{12}{c}}
v_{1}            & v_{2}            & \ldots & v_{n}\\
v_{1}\alpha_{1}  & v_{2}\alpha_{2}  & \ldots & v_{n}\alpha_{n}\\
\vdots           & \vdots           & \ddots & \vdots \\
v_{1}\alpha_{1}^{k-1}  & v_{2}\alpha_{2}^{k-1}  & \ldots & v_{n}\alpha_{n}^{k-1}
\end{array}
\right)
\end{equation}
and ${\GRS}_{n,k}(\al,\v)$ is an MDS code with parameters $[n,k,n-k+1]_{q}$. 
Hence, the following result gives a construction of an AQMDS code of length $n\le q$.

\begin{theorem}\label{thm:nestedGRS}
Let $q \geq 3$. Let $n$, $k$ and $j$ be positive integers such that 
$n \leq q$, $k \leq n-2$ and $j \leq n-k-1$. Then there exists a nontrivial AQMDS 
code with parameters $[[n,j,d_{z}/d_{x}]]_{q}$ where $\{d_{z},d_{x}\} = \{n-k-j+1,k+1\}$.
\end{theorem}
\begin{proof}Apply Theorem~\ref{thm:main} with $C_{1}^{\perp}=({\GRS}_{n,k}(\al,\v)) \subset C_{2}={\GRS}_{n,k+j}(\al,\v)$.
\end{proof}

\section{AQMDS Codes of Length $n=q+1$}\label{sec:n=q+1}
Let $\alpha_1,\alpha_2,\ldots,\alpha_q$ be distinct elements in ${\F}_q$ and 
$v_1,v_2,\ldots,v_{q+1}$ be nonzero 
elements in ${\F}_q$. Let $k\le q$ and consider the code $E$ given by
\begin{equation*}
E:= \left\{(v_{1}f(\alpha_{1}),\ldots,v_{q}f(\alpha_{q}),v_{q+1}f_{k-1}):f(X)=\sum^{k-1}_{i=0} f_{i}X^{i} \in {\F}_{q}[X]_{k}\right\}\text{.}
\end{equation*}
Let $\mathbf{x}=(0,\ldots,0,v_{q+1})$ and $G$ be as in (\ref{eq:G}) with $n=q$. 
Then $G_{E}:=\left( G|\mathbf{x}^{\mathrm{T}}\right)$ is a generator matrix of $E$. The code 
$E$ is an extended GRS code with parameters $[q+1,k,q-k+2]_{q}$ (see~\cite[Sect. 5.3]{HP03}).

Let $1\le r\le k-2$. Then there exists a monic irreducible polynomial $p(X) \in {\F}_{q}[X]$ 
of degree $k-r$ \cite[Cor. 2.11]{LN97}. By the choice of $p(X)$, observe that $p(\alpha_{i})\ne 0$ for all $i$. 
Hence, the matrix
\begin{equation}\label{eq:G_calC}
G_{C}=
{ 
\left(
\begin{array}{cccc}
v_{1} p(\alpha_{1})                                    & \ldots & v_{q} p(\alpha_{q}) & 0\\
v_{1} \alpha_{1} p(\alpha_{1})                         & \ldots & v_{q} \alpha_{q} p(\alpha_{q}) & 0\\
\vdots                                                 & \ddots & \vdots & \vdots \\
v_{1} \alpha_{1}^{r-2} p(\alpha_{1})                   & \ldots & v_{q} \alpha_{q}^{r-2} p(\alpha_{q}) & 0\\
v_{1} \alpha_{1}^{r-1} p(\alpha_{1})                   & \ldots & v_{q} \alpha_{q}^{r-1} p(\alpha_{q}) & v_{q+1}\\
\end{array}
\right)}
\end{equation}
is a generator matrix of a $[q+1,r,q-r+2]_{q}$-MDS code $C$.

Observe that, for all $g(X) \in {\F}_{q}[X]_{r}$, $p(X)g(X)$ is also a polynomial in ${\F}_{q}[X]_{k}$. 
Moreover, the coefficient of $X^{k-1}$ in $p(X)g(X)$ is given by the coefficient of $X^{r-1}$ in $g(X)$. 
Thus, $C \subset E$, leading to the following construction of AQMDS code of length $q+1$.

\begin{theorem}\label{thm:extGRS}
Let $q\ge 3$. Let  $j,k$ be positive integers such that $3\le k \le q$ and  $2 \le j \le k-1$. Then there exists an 
AQMDS code with parameters $[[q+1,j,d_{z}/d_{x}]]_{q}$ where $\{d_{z},d_{x}\}=\{q-k+2,k-j+1\}$.
\end{theorem}

\begin{proof}
Let $r=k-j$. Apply Theorem~\ref{thm:main} with $C_{1}=C^{\perp}$ 
and $C_{2}=E$.
\end{proof}

Note that Theorem~\ref{thm:extGRS} gives AQMDS codes with parameters $[[q+1,j,d_{z}/d_{x}]]_{q}$ with $j\ge 2$. 
The next proposition gives the necessary and sufficient conditions for the existence of 
pure CSS AQMDS codes with $j=1$.

\begin{proposition}\label{prop:pns}
Let $n,k$ be positive integers such that $k \le n-1$. There exists a pair of nested MDS 
codes $C \subset C^\prime$ 
with parameters $[n,k,n-k+1]_{q}$ and $[n,k+1,n-k]_{q}$, respectively, if and only if 
there exists an MDS code with parameters $[n+1,k+1,n-k+1]_{q}$. 

Equivalently, there exists a pure CSS AQMDS code with parameters $[[n,1,d_{z}/d_{x}]]_{q}$ 
where $\{d_{z},d_{x}\}=\{n-k,k+1\}$ if and only if there exists an MDS code with parameters 
$[n+1,k+1,n-k+1]_{q}$.
\end{proposition}
\begin{proof}
Let $G$ be a generator matrix of $C$. Pick $\mathbf{w} \in C^\prime \setminus C$ and observe that 
$\left(\begin{array}{c} G \\ \hline \mathbf{w} \end{array}\right)$ is a generator matrix for $C^\prime$. 
It can be verified that
\begin{equation*}
\left(\begin{array}{c|c} \0 & G \\ \hline 1 & \mathbf{w} \end{array}\right)
\end{equation*}
is a generator matrix of an $[n+1,k+1,n-k+1]_q$-MDS code.

Conversely, let $D$ be an $[n+1,k+1,n-k+1]_q$-MDS code 
with $k \le n-1$. Shortening the code $D$ at the last coordinate 
yields an $[n,k,n-k+1]_{q}$-MDS code $C$. Puncturing the code $D$ at the last coordinate gives 
an $[n,k+1,n-k]_{q}$-MDS code $C^\prime$. A quick observation confirms that $C \subset C^\prime$. \end{proof}

This proposition leads to the following characterization.

\begin{corollary}\label{cor:nonexist}
Assuming the validity of the MDS conjecture, there exists 
a pure CSS AQMDS code with parameters $[[q+1,1,d_{z}/d_{x}]]_{q}$ 
if and only if $q$ is even and $\{d_{z},d_{x}\}=\{3,q-1\}$.
\end{corollary}
\begin{proof}
There exists a $[2^{m}+2,3,2^{m}]_{2^{m}}$-MDS code (see~\cite[Ch. 11, Th. 10]{MS77}). 
By Proposition~\ref{prop:pns}, an AQMDS code with the indicated parameters exists.

The necessary condition follows from combining the MDS conjecture and Proposition~\ref{prop:pns}. 
Assume that there exists a $[[q+1,1,d_{z}/d_{x}]]_{q}$-AQMDS code $Q$ with $d_{x} \geq 2$. 
If $q$ is odd, the existence of $Q$ would imply the existence of a nontrivial MDS code of 
length $q+2$, contradicting the MDS conjecture. For even $q$, suppose $\{d_{z},d_{x}\} \neq \{q-1,3\}$. Without loss of generality,
assume $d_{z} \geq d_{x} \neq 3$. Then there exists a nested pair $[q+1,q+1-d_{x},d_{x}+1]_{q} 
\subset [q+1,q+2-d_{x},d_{x}]_{q}$. By Proposition~\ref{prop:pns}, there exists a 
$[q+2,q+2-d_{x},d_{x}+1]_{q}$-MDS code. If $d_{x}=2$, then $q+2-d_{x}=q \notin \{3,q-1\}$, 
contradicting the MDS conjecture. If $d_{x} > 3$, then $d_{z} < q-1$ and $3 < q+2-d_{z} \leq q+2-d_{x} < q-1$, 
a contradiction to the MDS conjecture.
\end{proof}

\section{AQMDS Codes of Length $n=2^{m}+2 \geq 6$ with $d_{z}=d_{x}=4$}\label{sec:n=q+2}
MDS codes of length $q+2$ are known to exist for $q=2^m$, 
and $k\in\{3,2^m-1\}$ \cite[Ch. 11, Th. 10]{MS77}.
Let $v_1,v_2,\ldots,v_{q+2}$ be nonzero elements in ${\F}_q$ and fix $\alpha_q=0$ in the notations of Section~\ref{sec:GRSuptoq}.

For $m\ge 2$, a generator matrix for $k=3$ or a parity check matrix for $k=2^{m}-1$ is given by
\begin{equation}\label{eq:paritycheck}
H=\left(
\begin{array}{cccccc}
v_1 			     		& \cdots        	&	v_{q-1}				& v_q  	& 0 			& 0\\
v_1\alpha_1 	        	& \cdots 		&	v_{q-1}\alpha_{q-1}		& 0 		& v_{q+1} 	& 0\\
v_1\alpha_1^2 	        	& \cdots 		&	v_{q-1}\alpha_{q-1}^2	& 0 		& 0 			& v_{q+2}
\end{array}
\right).
\end{equation}

Let $C$ be a $[2^{m}+2,2^{m}-1,4]_{2^{m}}$-code with parity check matrix $H$ 
given in (\ref{eq:paritycheck}). Let $D$ be the $[2^{m}+2,3,2^{m}]_{2^{m}}$-code 
whose generator matrix $G$ is given by
\begin{equation}\label{eq:genmatrix}
G=\left(
\begin{array}{cccccc}
v_1^{-1} 				       			& \cdots         &	v_{q-1}^{-1}				& v_q^{-1}  	& 0 			& 0\\
v_1^{-1}\alpha_1^{-1} 	            	& \cdots 		&	v_{q-1}^{-1}\alpha_{q-1}^{-1}	& 0 		& v_{q+1}^{-1} 	& 0\\
v_1^{-1}\alpha_1^{-2} 	        		& \cdots 		&	v_{q-1}^{-1}\alpha_{q-1}^{-2}	& 0 		& 0 			& v_{q+2}^{-1}
\end{array}
\right)\text{.}
\end{equation}
The following theorem gives a construction of an AQMDS code of length $q+2$.

\begin{theorem}\label{thm:n=q+2dx4}
Let $q=2^{m} \geq 4$. Then there exists an AQMDS code with 
parameters $[[2^{m}+2,2^{m}-4,4/4]]_{2^{m}}$.
\end{theorem}
\begin{proof}
First we prove that $D \subset C$ by showing that 
$M=(m_{i,j}):=GH^{T}=\0$. Note that
\begin{equation*}
m_{i,j}=\sum_{l=1}^{q+2}g_{i,l} \cdot h_{j,l}
\end{equation*}
for $1 \leq i,j \leq 3$. If $i=j$, then $m_{i,j}=q=0$. If $i \neq j$, 
the desired conclusion follows since
\begin{equation*}
\sum_{i=1}^{q-1} \alpha_i  =\sum_{i=1}^{q-1} \alpha_i^{-1} =0 \text{ and } 
\sum_{i=1}^{q-1} \alpha_i^{-2}=\sum_{i=1}^{q-1} \alpha_i^2 =\left(\sum_{i=1}^{q-1} \alpha_i\right)^2 =0\text{.}
\end{equation*}
Applying Theorem~\ref{thm:main} with $C_{1}=D^{\perp}$ 
and $C_{2}=C$ completes the proof.
\end{proof}

\section{AQMDS Codes with $d_{z} \geq d_{x}=2$, an Alternative Look}\label{sec:dxequals2}
In the previous sections, suitable pairs of GRS or extended GRS codes were chosen for the CSS construction. 
This section singles out the case of $d_{x}=2$ where the 
particular type of the MDS code chosen is inessential. The following theorem gives a construction on an AQC with $d_x=2$.

\begin{theorem}[{\cite[Th. 7]{EJL2011}}]\label{thm:spectrumn}
Let $C$ be a linear (not necessarily MDS) $[n,k,d]_{q}$-code with $k\ge 2$. If $C$ has a 
codeword $\u$ such that $\wt(\u)=n$, then there exists an $[[n,k-1,d/2]]_{q}$-AQC.
\end{theorem}

Let $C$ be an $[n,k,n-k+1]_{q}$-MDS code. Ezerman \etal~\cite{EGS09} showed that 
$C$ has a codeword $\u$ with $\wt(\u)=n$, except when either $C$ is the dual of the binary 
repetition code of odd length $n\ge 3$, or $C$ is a simplex code with parameters $[q+1,2,q]_q$. 
Hence, the following corollary can be derived.

\begin{corollary}\label{cor:MDS}
The following statements hold:
\begin{enumerate}
	\item For even integers $n$, there exists an $[[n,n-2,2/2]]_{2}$-AQMDS code.
	\item For positive integers $n,q \geq 3$, there exists an $[[n,n-2,2/2]]_{q}$-AQMDS code.
	\item Given positive integers $q \geq n \ge 4$, there exists an AQMDS code for $2 \le k \leq n-2$ 
	with parameters $[[n,k-1,d_{z}/2]]_{q}$ with $d_{z}=n-k+1$.
	\item Given $q \geq 4$, there exists an AQMDS code for $3 \leq k \leq q-1$ with parameters 
	$[[q+1,k-1,d_{z}/2]]_{q}$ with $d_{z}=q-k+2$.
	\item Given positive integer $m\ge 2$ and $q=2^{m}$, there exists an AQMDS code with parameters $[[2^{m}+2,2,2^{m}/2]]_{2^{m}}$ and an AQMDS code with parameters $[[2^{m}+2,2^{m}-2,4/2]]_{2^{m}}$.
\end{enumerate}
\end{corollary}

Wang \etal~\cite[Cor. 3.4]{WFLX09} gave a different proof of the existence of $[[n,n-2,2/2]]_{q}$-AQMDS codes $Q$ for $n,q \geq 3$.

In this section, it is shown for $d_x=2$ that the specific construction of the classical MDS codes used in the CSS 
construction is inconsequential. This is useful as there are many classical MDS codes which are not equivalent to 
the GRS codes (see~\cite{RL89}, for instance). 

\section{Summary}\label{sec:summary}
While the ingredients to construct a pure AQC under the CSS construction, namely a pair of nested codes, the knowledge 
on the codimension and the dual distances of the codes, are all classical, computing the exact set of parameters 
and establishing the optimality of the resulting AQC are by no means trivial. 

This work shows how to utilize the wealth of knowledge available regarding classical MDS codes to completely classify 
under which conditions there exists a particular pure CSS AQMDS code and how to construct such a code explicitly. 
Outside the MDS framework, more work needs to be done in determining the exact values of $d_x$ and $d_z$ and in establishing
optimality.

We summarize the results of the paper in the following theorem.

\begin{theorem}\label{thm:summary}
Let $q$ be a prime power, $n, k$ be positive integers and $j$ be a nonnegative integer. Assuming the validity 
of the MDS conjecture, there exists a pure CSS AQMDS code with parameters $[[n,j,d_z/d_x]]_q$, where
$\{d_z,d_x \}=\{n-k-j+1,k+1\}$ if and only 
if one of the following holds:

\begin{enumerate}
\item\emph{[Prop.~\ref{prop:dxis1}, Prop.~\ref{prop:kis0}]} $q$ is arbitrary, $n\ge 2$, $k\in\{1,n-1\}$, and $j\in\{0,n-k\}$.
\item\emph{[Cor.~\ref{cor:MDS}] }$q=2$, $n$ is even, $k=1$, and $j=n-2$.
\item\emph{[Cor.~\ref{cor:MDS}]} $q\ge 3$, $n\ge 2$, $k=1$, and $j=n-2$.
\item\emph{[Prop.~\ref{prop:dxis1}, Prop.~\ref{prop:kis0}, Th.~\ref{thm:nestedGRS}]} $q\ge 3$, $2\le n\le q$, $k \le n-1$, and $0\le j\le n-k$.
\item\emph{[Prop.~\ref{prop:dxis1}, Prop.~\ref{prop:kis0}, Th.~\ref{thm:extGRS}]} $q\ge 3$, $n=q+1$,  $k \le n-1$, and $j\in\{0,2,\ldots,n-k\}$.
\item\emph{[Cor.~\ref{cor:nonexist}]} $q=2^m$, $n=q+1$, $j=1$, and $k\in \{2,2^{m}-2 \}$.

\item\emph{[Prop.~\ref{prop:dxis1}, Prop.~\ref{prop:kis0}, Th.~\ref{thm:n=q+2dx4}, Cor.~\ref{cor:MDS}]} $q=2^m$ where $m\ge 2$, $n=q+2$,  
$$ \left\{
\begin{array}{l} k=1,\mbox{ and }j\in\{2,2^m-2\},\\
k=3,\mbox{ and }j\in\{0,2^m-4,2^m-1\},\mbox{ or,}\\
k=2^m-1,\mbox{ and }j\in\{0,3\}.\end{array} \right. $$
\end{enumerate}
\end{theorem}

As a concluding remark, note that all AQMDS codes constructed here are pure CSS codes. The existence of a 
degenerate CSS AQMDS code or an AQMDS code derived from non-CSS method with parameters different from those 
in Theorem~\ref{thm:summary} remains an open question.

\section*{Acknowledgments}
The authors thank Markus Grassl for useful discussions and for suggesting Proposition~\ref{prop:pns}. 
The work of S.~Jitman was supported by the Institute for the Promotion of Teaching Science 
and Technology of Thailand. The work of all of the authors is partially supported by Singapore National 
Research Foundation Competitive Research Program Grant NRF-CRP2-2007-03.

\vspace*{-6pt}   

%

\end{document}